\documentclass[12pt,a4paper]{amsart}
\usepackage{amsfonts}
\usepackage{amsthm}
\usepackage{amssymb, amsmath}
\usepackage{amscd}
\usepackage[utf8]{inputenc}
\usepackage{bbold}
\usepackage{t1enc}
\usepackage[mathscr]{eucal}
\usepackage{indentfirst}
\usepackage{graphicx}
\usepackage{graphics}
\usepackage{pict2e}

\numberwithin{equation}{section}

\theoremstyle{plain}
\newtheorem{Th}{Theorem}[section]

\newtheorem{Pro}[Th]{Proposition}

 \theoremstyle{definition}
\newtheorem{Def}[Th]{Definition}

\newtheorem{Rem}[Th]{Remark}
\newtheorem{?}[Th]{Problem}

\DeclareMathOperator{\signdim}{sign.dim}
\DeclareMathOperator{\supp}{supp}

\newcommand{\Prob}{\mathbb{P}}
\newcommand{\R}{\mathbb{R}}

\newcommand{\F}{\mathbb{F}}

\begin{document}

\title{On the signaling dimension of the cross-polytope}

 \author[T.\ Bencze, P.\ E.\ Frenkel]{Tam\'as Bencze and P\'eter E.\ Frenkel}

\address{E\"{o}tv\"{o}s Lor\'{a}nd University,
  P\'{a}zm\'{a}ny P\'{e}ter s\'{e}t\'{a}ny 1/C, Budapest, 1117 Hungary  and R\'enyi Institute,  Budapest, Re\'altanoda u.\ 13-15, 1053 Hungary
  }
\email{ benczetamas11@gmail.com, peter.e.frenkel@gmail.com}

\begin{abstract}In the general probabilistic model of information theory, the signaling dimension of  a given communication protocol (or, equivalently, of  its state space) is the  smallest possible complexity of a  classical communication protocol capable of perfectly simulating the given one in terms of input-output correlations. 

The signaling dimensions of three-, four-, and five-dimensional cross-polytopes have been known to be 3.  We give logarithmic lower and upper bounds for the signaling dimensions of higher dimensional cross-polytopes, pinning down the value up to an additive gap of 2. The lower bound shows that the signaling dimension goes to infinity, in contrast to the case of balls and hypercubes.
\end{abstract}

\maketitle
\section{Preliminaries}
We recall the  notions and basic facts related to the general probabilistic model of information theory and to  signaling dimension, cf.\ \cite{A, ATB, F, FW}.

In this Introduction, the letters $K$ and $L$ will always denote convex bodies (i.e., convex, compact sets with non-empty interior) in finite-dimensional real affine spaces. The \emph{affine dimension} of a convex body is the dimension of the ambient affine space. Examples of convex bodies of affine dimension $d$ are the hypercube $[-1,1]^d$ and the cross-polytope $$\diamond_d=\{(x_1, \dots, x_d)\in\R^d:|x_1|+\dots+|x_d|\le 1\}.$$
\begin{Def}
An \emph{effect} on $K\subset\R^d$ is an affine function $\mathbb{R}^d\rightarrow\mathbb{R}$ that is nonnegative on $K$. The \emph{effect cone} is the set of effects.
\end{Def}
\begin{Def}
A \emph{measurement} on $K$ is a finite collection of effects on $K$ summing  to the constant 1 function.
A measurement is \emph{primitive} if it consists of linearly independent \emph{extremal} effects (i.e., effects  belonging to extremal rays of the effect cone).
\end{Def}
\begin{Rem}\label{extr}
The  $k$-part measurements on $K$ form a convex compact set. Every measurement can be obtained from primitive measurements in finitely many steps by taking convex combinations of measurements, replacing two effects with their sum, and inserting zero effects.  This follows from~\cite[Theorem 2]{A} or, alternatively,  from the closely related~\cite[Proposition 1.1]{F}.
\end{Rem}
\begin{Def}
For (extremal) points $p_1$, $p_2$, \dots, $p_l$  of $K$ and a (primitive) measurement $(e_1,e_2, \dots,e_k)$ on $K$, we call the matrix with $(i,j)$-entry $e_i(p_j)$ a \emph{(primitive) transition matrix} of $K$.
\end{Def}

\begin{Rem}\label{transextr} A transition matrix is always \emph{stochastic}, i.e., all entries are nonnegative and each column sums to 1.
The  set of $k\times l$ transition matrices of $K$ is, in general, \emph{not} convex. 

Every transition matrix can be obtained from primitive transition matrices  in finitely many steps by taking convex combinations of transition matrices, replacing two rows with their sum,  inserting zero rows,  inserting convex combinations of existing columns, and deleting columns. This follows from Remark~\ref{extr}.
\end{Rem}

Convex combinations of transition matrices arise naturally in information theory because they correspond to the sender and the receiver having access to shared randomness.

\begin{Def}
$K$ can be simulated by $L$ if every transition matrix of $K$ can be written as a convex combination of transition matrices of $L$.
\end{Def}
\begin{Rem}\label{primsuff} (a) If every \emph{primitive} transition matrix of $K$ can be written as a convex combination of transition matrices of $L$, then $K$ can be simulated by $L$.  This follows from Remark~\ref{transextr}.

 (b)   Simulability is reflexive and transitive. 
    
(c)    If $K$ is an affine image of $L$, then $K$ can be simulated by $L$.
\end{Rem}
\begin{Def}\label{signdim}~\cite{A}
The \emph{signaling dimension} of $K$, sometimes denoted by $\signdim K$, is the smallest $r$ such that $K$ can be simulated by a simplex with $r$ vertices.
\end{Def}
Note that the simplex with $r$ vertices is the space of mixed states corresponding to classical communication with $r$ pure states.
Definition~\ref{signdim} was motivated by the main  result of~\cite{FW}: the signaling dimension of the set of $r$-square density matrices (i.e.\ quantum states) is $r$.

We shall make use of the following description of classical simulability.

\begin{Rem}\label{simplex}\cite{FW} For a  matrix $M$, the following are equivalent.

(a) $M$ is a convex combination of  transition matrices of the simplex with $r$ vertices.

(b) $M$  can be written as a convex combination of stochastic matrices with at most $r$ nonzero rows.
\end{Rem}

The following proposition is almost certainly known, but we include a proof for the convenience of the reader. 

\begin{Pro}\label{extreff}
Assume that  the convex body $K$ is a polytope. 

(a) For each facet $S$ of $K$, there is a nonzero effect on $K$ vanishing on $S$, and it is unique up to a constant factor. 

(b) Such effects are extremal, 
and these are the only extremal effects on $K$.
\end{Pro}
\begin{proof}
(a) Let $P\in K$ be a point outside $S$. For any fixed positive value $c$, there is a unique affine function $e$ that vanishes on $S$ and has $e(P)=c$. This $e$ is nonnegative on $K$. If we fix another value $\lambda c$, then the effect we get is $\lambda e$. 

(b) First observe that the vanishing set of an effect $e$ is a face of $K$, so it is either contained in a facet, or  $e$ is constant $0$. 

If $e$ vanishes on a facet, then we can only write it as a positive combination of effects that also vanish on that facet, meaning only its scalar multiples, so $e$ is extremal. 

Otherwise, we can choose an effect  $g$ with a strictly larger vanishing set,  so, for small enough $\lambda>0$, we have  $e\geq\lambda g$ on the whole $K$, hence $e$ is not extremal: it is the sum of the effects $\lambda g$ and $e-\lambda g$, where $\lambda g$ is not a scalar multiple of $e$.
\end{proof}

Computing, or even efficiently estimating  the signaling dimension of a given 
state space is usually difficult. In~\cite{F}, it was shown that balls (or, equivalently, ellipsoids) of arbitrary affine dimension have signaling dimension 2. 
Hypercubes are easy:

\begin{Pro}
    For any $d\ge 1$, the hypercube $[-1,1]^d$ has signaling dimension $2$.
\end{Pro}

\begin{proof}
From Proposition~\ref{extreff}, it is easy to see that each primitive measurement has two parts; it consists of extremal effects vanishing on opposite facets. Therefore,  the signaling dimension is $\le 2$.  The reverse inequality is obvious as the hypercube is not just one point.
\end{proof}

\section{Bounding the signaling dimension of the cross-polytope}
The signaling dimensions of the cross-polytopes $\diamond_d$ for $d=1,2,3,4,5$ are $2,2,3,3,3$, respectively. Case $d=1$  is trivial; case $d=2$  was computed in~\cite{A}, case $d=3$  in~\cite{F}, and cases $d=4,5$  in~\cite{AK}, where the case of general  $d$ was posed as an open problem.  

From Remark~\ref{primsuff}(c), we have $\signdim\diamond_d\le\signdim\diamond_{d+1}$. 

By Proposition~\ref{extreff}, we can think of an extremal effect on $\diamond_d$ as a pair consisting of a facet and a number (the value of the effect on the opposite facet). Thus, we can think of a primitive measurement as an appropriate function from the set of facets to $[0,1]$. 

Alternatively, instead of the facets, we can consider the vertices of the dual polytope, which is the hypercube $[-1,1]^d$. In this terminology, a vertex $x\in\{-1,1\}^d$ gives rise to the extremal effect $e_x(v)=(1+x\cdot v)/2$, so that $e_x=1$ on the facet of $\diamond_d$ represented by $x$, and $e_x=0$ on the opposite facet. 

A function $f:\{-1,1\}^d\to [0,1]$ represents a (not necessarily primitive) measurement $\mathcal M_f=(f(x)e_x)_{x\in\supp f}$ if and only if $\sum_xf(x)e_x=1$, i.e, for any facet $H$ of the hypercube we have \begin{equation}\label{balanced}\sum\limits_{x \in H\cap\{-1,1\}^d}f(x)=1\end{equation}
The measurement  $\mathcal M_f$ is primitive if and only if the extremal effects $e_x$  $(x\in\supp f)$  are linearly independent, i.e., the support of $f$ is an affine independent set. 

Since we want to check simulability, we may assume that the states on which we use our measurement  are precisely  the $2d$ vertices of $\diamond_d$; cf.\ Remark
~\ref{primsuff}. In this way, the functions above also correspond to transition matrices: to a function $f$ corresponds the transition matrix $M_f=(f(x)e_x(p))_{x,p}$, where $x$ runs through the elements of $\supp f$ and  $p$ runs through the vertices of $\diamond_d$. 

In the sequel, we shall write  $\binom{\supp f}{\le r}$ for the family of those subsets of $\supp f$ that have cardinality $\le r$.  Furthermore, for any set $S\subseteq \{-1,1\}^d$, we write $f(S)=\sum_{x\in S}f(x)$

\begin{Pro}\label{model} Let  $f:\{-1,1\}^d\to [0,1]$  satisfy the equations~\eqref{balanced}.  Then the following are equivalent.
 
  (a) The transition matrix $M_f$ 
  can be written as a convex combination of transition matrices of the simplex with $r$ vertices. 
  
  (b) There exists a 
  probability distribution $F$ on the family $\binom{\supp f}{\le r}$ 
  such that for any subset $S\subseteq\supp f$ contained in any facet of $[-1,1]^d$, we have \begin{equation}\label{Hall}f(S)\leq \sum_{
  X\cap S\neq\emptyset}
  F(X). \end{equation}
    
\end{Pro}
\begin{proof}
By Remark~\ref{simplex}, (a) holds 
if and only if there is a  probability distribution $F$ on the family $\binom{\supp f}{\le r}$, 
and a stochastic matrix $A(X)$ for every  
subset $X\in\binom{\supp f}{\le r}$, which  has only zero entries outside the rows labeled by elements of $X$, such that $$\sum_X F(X)A(X)=M_f.$$ If we fix $F$, we can separate the conditions by column: for each vertex $p$ of $\diamond_d$, the column of $A(X)$ labeled with $p$ is an arbitrary probability distribution on $X$. For a fixed $p$, the existence of probability distributions $A(X)_{\cdot,p}$ with the condition $$\forall x \qquad\sum_X F(X)A(X)_{x,p}=f(x)e_x(p)$$ can be decided using the Supply--Demand Theorem \cite[2.1.5. Corollary]{LoPlu} on the  bipartite  incidence graph between the elements $x$ of $\supp f$ 
and the $r$-element subsets $X$ of $\supp f$, with weights $f(x)e_x(p)$ on the elements and with weights $F(X)$ on the subsets. The theorem says that there exists a probability distribution on the edges with these marginals --- which is equivalent to the existence of the required probability distributions $A(X)_{.,p}$ --- if and only if for each subset $S$ of $\supp f$, the inequality \begin{equation}\label{Hallp}\sum\limits_{x\in S}f(x)e_x(p)\leq \sum_{
X\cap S\neq\emptyset}
F(X) \end{equation} holds. But it is enough to check \eqref{Hallp} for subsets $S$ on the facet $\{y|e_y(p)=1\}$, so we get the inequality~\eqref{Hall}  for this facet. Now taking all the constraints from $p$ running through all vertices of $\diamond_d$, we get (b).
\end{proof}

\subsection{Lower bound}
\begin{Th}\label{lowerbound} 
The signaling dimension of the $d$ dimensional cross-polytope is at least $\lceil\log_2(d+2)\rceil$.
\end{Th}
\begin{proof}Set $r=\lceil\log_2(d+2)\rceil-1=\lfloor\log_2(d+1)\rfloor$. We aim to prove that $\signdim\diamond_d>r$.  We have $2^r\le d+1$, and, from Remark~\ref{primsuff}(c),  we may assume that this holds with equality.

Recall the construction of the \emph{shortened Hadamard code} or \emph{simplex code} $W<\F_2^d$. It is the row space of the matrix $A\in\F_2^{r\times d}$ whose columns are all the nonzero vectors in $\F_2^r$, in some order.\footnote{Note that $A$ is the parity check matrix of the \emph{Hamming code}, which is therefore the dual of the shortened Hadamard code.} Since $A$ has rank $r$, we have $|W|=2^r$. Since $A$ has no zero column and $W$ is a linear subspace, $W$ has $2^{r-1}$ points on any coordinate hyperplane.

Any $r$ elements of $W$  are the rows of a  matrix $BA$ with $B\in M_r(\F_2)$. The $r$ rows of $B$ are always on an affine  hyperplane in $\F_2^r$ and therefore any $r$ elements of $W$ have a common coordinate, i.e., in  the natural identification $\F_2^d\leftrightarrow \{-1,1\}^d$ they are on a  facet.

   Put weight $2^{1-r}$ on each point of $W$  and consider the natural identification above 
   to get a function $f:\{-1,1\}^d\to \{0, 2^{1-r}\}$ satisfying the equations~\eqref{balanced}.

   By Proposition~\ref{model}, it suffices to prove that no probability distribution $F$ on $\binom{\supp f}{\le r}$ 
   can satisfy the inequalities~\eqref{Hall} for all subsets $S\subseteq\supp f$ belonging to a facet of the hypercube $[-1,1]^d$. And indeed, for any distribution $F$ we have  $F(X)>0$  for some subset $X \in\binom{\supp f}{\le r}$. But $X\cap W$ belongs to a facet $H$, so $S:=W\setminus H$  belongs to the opposite facet  and is disjoint from $X$. Therefore,  \eqref{Hall} would imply $1\le 1-F(X)$, a  contradiction.
\end{proof}

\subsection{Upper bound}  We start with an elementary bound which is only useful for small values of $d$.
\begin{Th}
The signaling dimension of the $d$ dimensional cross-polytope is at most $\lfloor(d+3)/2\rfloor$.
\end{Th}
\begin{proof}
Let $r=\lfloor(d+3)/2\rfloor$. By Remark~\ref{primsuff} and Definition~\ref{signdim}, the signaling dimension of $\diamond_d$ is at most $r$ if every primitive transition matrix of $\diamond_d$ can be written as a convex combination of transition matrices of the simplex with $r$ vertices. To prove this, we will use Proposition~\ref{model}, so we take a primitive transition matrix defined by a weight function $f$. 
Then supp($f$) is affine independent, hence it has at most $d+1$ points, so there exists a facet $H$ of $[-1,1]^d$ on which supp($f$) has at most $({d+1})/{2}$ points. 
Let $L=H\cap \supp f$.

Now we define the function $F$: set  $F(L\cup \{x\})=f(x)$ for each $x\in \supp f\setminus H$, and let $F= 0 $ everywhere else. Note that $$|L\cup \{x\}|\le \left\lfloor\frac{d+1}{2}\right\rfloor+1=r.$$ Now we only have to check that $$\sum_{x\in S}f(x)\leq \sum_{
X\cap S\neq\emptyset}F(X)=\sum\limits_{(L\cup\{x\})\cap S\neq\emptyset}F(L\cup\{x\}) $$ holds for each subset  $S$ of some facet of $[-1,1]^d$. If $S\cap L\neq \emptyset$, then the right hand side is $1$, while the left hand side is at most $1$. Otherwise, both sides are equal to $\sum_{x\in S}F(L\cup\{x\})$.
\end{proof}

We now pass to a more involved result which gives  a much tighter upper bound when $d$ is large; in fact, together with Theorem~\ref{lowerbound}, it pins down the signaling dimension up to an additive gap of 2. 

\begin{Th}\label{thm:logarithmic}
 For integers \(r\geq 6\) and $1\le d\le 2^{r-2}$, one has
$$
 \signdim\diamond_d\leq r.$$
Equivalently,  for all $d\ge 1$,
\[
 \signdim\diamond_d
 \leq \max\{4,\lceil\log_2 d\rceil\}+2.
\]
\end{Th}

\begin{proof}
Fix a measurement function \(f\) as above, and let
\(\mu(x)=f(x)/2\). Thus, \(\mu\) is a probability distribution
and \(\mu(H)=1/2\) for each facet $H$ of the hypercube.
Let \(Z_1,\ldots,Z_r\) be independent samples with law \(\mu\).
Let \(E\) be the event that the samples meet all facets. 
Put
\[
 b=2^{-r},\qquad \varepsilon=2db,\qquad q=\Prob(E).
\]
Each event that all samples lie on a specified facet 
has probability \(b\). Consequently, the union bound gives
\[
 q\geq 1-\varepsilon\geq\frac12.
\]
Also \(q\leq 1-b\), since \(E\) excludes the event that all
samples lie in any one fixed facet. 

Define \(F\) to be the conditional law, given \(E\), of the set
\[
 X=\{Z_1,\ldots,Z_r\}.
\]
Repeated samples are allowed, so \(|X|\leq r\).
We verify inequality~\eqref{Hall}.

Fix \(S\subseteq\operatorname{supp}f\) contained in a facet, 
and put \(t=f(S)\). Then \(0\leq t\leq 1\).
Let \(A\) be the event that at least one sample belongs to \(S\).
Its unconditional probability is
\[
 h(t):=\Prob(A)=1-(1-t/2)^r.
\]

First we bound the probability that \(A\) occurs but \(E\)
fails. For a fixed facet $H$,  put
\(a=f(S\cap H)\), so that
\(f(S\cap H')=t-a\)  for the opposite facet $H'$.
The probability of \(A\) together with all samples lying in
one of these two facets is
\[
 b\bigl[1-(1-a)^r+1-(1-t+a)^r\bigr].
\]
The function \(u\mapsto 1-(1-u)^r\) is concave on \([0,1]\).
Hence the last expression is at most
\[
 2b\bigl[1-(1-t/2)^r\bigr]=2b\,h(t).
\]
Summing over the \(d\) pairs of opposite facets gives
\begin{equation}\label{eq:hit}
 \Prob(A\cap E)\geq (1-\varepsilon)h(t)\geq \frac12 h(t).
\end{equation}

For a second estimate, choose a facet containing
\(S\). The event that all samples lie on the opposite facet has
probability \(b\), misses \(S\), and is disjoint from \(E\).
It follows that
\begin{equation}\label{eq:miss}
 \Prob(A^c\cap E)\leq (1-t/2)^r-b=:g(t).
\end{equation}

We shall use the elementary inequality
\begin{equation}\label{eq:scalar}
 (5/6)^r\leq \frac13+\frac23\,2^{-r}
             =\frac13+\frac23 b
 \qquad (r\geq 6).
\end{equation}
For \(r=6\), this is
\(15625/46656\leq 11/32\).
For \(r\geq 7\), it follows from
\((5/6)^r\leq (5/6)^7<1/3\).

Suppose first that \(0\leq t\leq 1/3\).
Since \(h\) is concave and \(h(0)=0\), inequality
\eqref{eq:scalar} implies
\[
 h(t)\geq 3t\,h(1/3)
       =3t\bigl[1-(5/6)^r\bigr]
       \geq 2(1-b)t.
\]
Together with \eqref{eq:hit} and \(q\leq 1-b\), this yields
\[
 \Prob(A\cap E)\geq (1-b)t\geq qt.
\]

Now suppose that \(1/3\leq t\leq 1\).
The function \(g\) is convex, \(g(1)=0\), and
\eqref{eq:scalar} gives
\[
 g(1/3)=(5/6)^r-b
       \leq \frac13-\frac13 b\leq\frac13.
\]
By convexity, its graph on \([1/3,1]\) lies below the chord
joining its endpoint values. Therefore
\[
 g(t)\leq \frac{1-t}{1-1/3}\,g(1/3)
       \leq \frac{1-t}{2}
       \leq q(1-t).
\]
Using \eqref{eq:miss}, we obtain
\[
 \Prob(A^c\mid E)\leq 1-t.
\]

In both cases, we have
\[
 \Prob_F(X\cap S\neq\varnothing)
 =\Prob(A\mid E)\geq t=f(S),
\]  verifying~\eqref{Hall}. An application of Proposition~\ref{model}  concludes the proof. 
\end{proof}

\section*{Acknowledgments and AI declaration} The second author's research has been partially supported by NKFIH grants K 146380, KKP 139502, and Excellence  151342. 

Theorem~\ref{thm:logarithmic} and its proof were found by ChatGPT 6.1-Sol after all the other results had been discovered  and proved by the human authors.

\end{document}